\documentclass[11pt,letterpaper,twoside,english]{article}
\usepackage{fullpage} 
\usepackage{mdwlist}
\usepackage{amssymb,amsbsy,latexsym}
\usepackage{amsmath}
\usepackage{graphics, subfigure, float}
\usepackage[applemac]{inputenc}
\usepackage[american]{babel}
\usepackage[T1]{fontenc} 
\usepackage{palatino}

\usepackage{amscd,amsthm}

\usepackage{verbatim, comment}
\usepackage[dvips]{graphicx,epsfig,color}
\usepackage{pst-all}
\usepackage{pstricks-add}
\newpsobject{showgrid}{psgrid}{subgriddiv=1,griddots=10,gridlabels=6pt}

\newtheoremstyle{theorem}{1em}{1em}{\slshape}{0pt}{\bfseries}{.}{ }{}
\theoremstyle{theorem}
\newtheorem{theorem}{Theorem}
\newtheorem*{theorem*}{Theorem}
\newtheorem{corollary}[theorem]{Corollary}

\newtheorem{lemma}[theorem]{Lemma}

\newtheorem*{conjecture*}{Conjecture}
\newtheorem*{problem*}{Problem}

\theoremstyle{remark}

\newtheorem*{remark*}{Remark}

\newtheorem*{algorithm*}{Algorithm}

\providecommand{\N}{\mathbb{N}}
\providecommand{\Z}{\mathbb{Z}}
\providecommand{\Q}{\mathbb{Q}}

\newcommand{\maxfs}{\textsc{MaxFS}}

\floatstyle{ruled}
\newfloat{Figure}{t}{fig}
\newfloat{algorithm}{tbp}{loa}
\floatname{algorithm}{Algorithm}
\newfloat{algorithm*}{tbp}{loa}
\floatname{algorithm*}{Algorithm}

\usepackage[displaymath,textmath,sections,graphics, subfigure, floats]{preview} 
\PreviewEnvironment{enumerate} 
\PreviewEnvironment{itemize}
\PreviewEnvironment{theorem} 
\PreviewEnvironment{lemma} 
\PreviewEnvironment{myproblem}
\PreviewEnvironment{corollary}
\PreviewEnvironment{abstract} 
\PreviewEnvironment{defn} 
\PreviewEnvironment{center} 
\PreviewEnvironment{pspicture} 

\newcommand{\budget}{b}

\newcommand{\eps}{\varepsilon}
\def\inlineshrink{-2mm} 

\psset{arrowsize=6pt, labelsep=2pt, linewidth=1.5pt}
\makeatother

\begin{document}

\date{}

\title{Prizing on Paths:\\A PTAS for the Highway Problem}

\author{Fabrizio Grandoni\thanks{Computer Science Department,
University of Rome Tor Vergata, Roma, Italy, {\tt grandoni@disp.uniroma2.it}. Developed while visiting EPFL.} \and 
Thomas Rothvo\ss\thanks{Institute of Mathematics, EPFL, Lausanne, Switzerland, \tt{thomas.rothvoss@epfl.ch}} }


\maketitle

\begin{abstract}
\noindent In the \emph{highway} problem, we are given an $n$-edge line graph (the highway), and a set of paths (the drivers), each one with its own budget. For a given assignment of edge weights (the tolls), the highway owner collects from each driver the weight of the associated path, when it does not exceed the budget of the driver, and zero otherwise. The goal is choosing weights so as to maximize the profit. 
A lot of research has been devoted to this apparently simple problem.
The highway problem was shown to be strongly $\mathbf{NP}$-hard only recently [Elbassioni,Raman,Ray-'09]. The best-known approximation is $O(\log n/\log\log n)$ [Gamzu,Segev-'10], which improves on the previous-best $O(\log n)$ approximation [Balcan,Blum-'06]. Finding a constant (or better!) approximation algorithm is a well-known open problem in network design. Better approximations are known only for a number of special cases.

In this paper we present a PTAS for the highway problem, hence closing the complexity status of the problem. Our result is based on a novel randomized dissection approach, which has some points in common with 
Arora's quadtree dissection for Euclidean network design [Arora-'98]. The basic idea is enclosing the highway in a bounding path, such that both the size of the bounding path and the position of the highway in it are random variables. Then we consider a recursive $O(1)$-ary dissection of the bounding path, in subpaths of uniform optimal weight. Since the optimal weights are unknown, we construct the dissection in a bottom-up fashion via dynamic programming, while computing the approximate solution at the same time. Our algorithm can be easily derandomized.

We demonstrate the versatility of our technique by presenting PTASs for two variants of the highway problem: the \emph{tollbooth} problem with a constant number of leaves and the \emph{maximum-feasibility subsystem} problem on interval matrices. In both cases the previous best approximation factors are polylogarithmic [Gamzu,Segev-'10,Elbassioni,Raman,Ray,Sitters-'09].
\end{abstract}

\setcounter{page}{0}
\thispagestyle{empty}

\newpage

\section{Introduction}

Consider the following setting. We are given a single-road highway, which is partitioned into segments by tollbooths. The highway owner fixes a toll for each segment. A driver traveling between two tollbooths pays the total toll of the corresponding segments. However, if the total toll exceeds the budget of the driver, she will not use the highway (e.g., she will take a plane). Our goal is maximizing the profit of the highway owner. To that aim, we need to compromise between very low tolls (in which case all the drivers take the highway, but providing a small profit) and very high tolls (in which case no driver takes the highway, and the profit is zero). It is not hard to imagine other applications with a similar nature. For example, the highway segments might be replaced by the links of a (high-bandwidth) telecommunication network. 

The \emph{highway} problem formalizes the scenarios above. We are given an $n$-edge line graph $G=(V,E)$ (the \emph{highway}), and a set $\mathcal{D}=\{D_1,\ldots,D_m\}$ of $m$ paths in $G$ (the \emph{drivers}), each one characterized by a value $\budget_j\in \Q_{\geq 0}$ (the \emph{budgets}). For a given weight function $w:E\to \Q_{\geq0}$ (the \emph{tolls}) and a driver $D$, let $w(D):=\sum_{e\in D}w(e)$ be the \emph{weight} of $D$\footnote{Throughout this paper we confuse graphs with their set of edges: the meaning will be clear from the context.}. Our goal is choosing $w$ so as to maximize the following \emph{profit} function: \vspace{\inlineshrink}
$$
\sum_{j: w(D_j)\leq \budget_j}w(D_j). \vspace{\inlineshrink}
$$



Despite the simplicity of its formulation and its clear relation to applications, there is a huge gap between known approximation and inapproximability results for the highway problem. 
The problem was shown to be strongly $\mathbf{NP}$-hard very recently \cite{ERR09corr}. The best-known approximation factor is $O(\log n/\log\log n)$ \cite{GS10} (see also \cite{BB06ec}). A quasi-polynomial-time approximation scheme (QPTAS) is given in \cite{ESZ07esa}. This is a strong evidence of the existence of a PTAS for the problem. However, even finding a constant approximation is considered a challenging open problem in network design. For this reason, researchers focused on some relevant special cases \cite{BB06ec,BK06soda,GHKKKM05soda,HK05wads}.


\subsection{Our Results and Techniques}\label{sec:results}

In this paper we present a deterministic polynomial-time approximation scheme (PTAS) for the highway problem, hence closing the complexity status of the problem. To achieve our goal, we exploit a novel randomized dissection approach. 

The basic idea is as follows. Let $\eps>0$ be a small constant. Via simple reductions (see Section \ref{sec:preliminaries}), we can restrict ourselves to the case that optimal weights $w^*(e)$ are in $\{0,1\}$, and that the sum $W^*$ of the optimal weights along the highway is polynomially bounded in the number $n$ of edges. This introduces a $1-\Theta(\eps)$ factor in the approximation. 

The dynamic program is based on the following strategy. 
We consider all the subpaths $P$ of the highway, and guess the value $W\in \{0,1,\ldots,W^*\}$ of the sum of the optimal weights along $P$. Note that the number of pairs $(P,W)$ is polynomially bounded in $n$, due to the reductions above. 

We next restrict our attention to the drivers $\mathcal{D}(P):=\{D\in \mathcal{D}:\;D\subseteq P\}$ which are entirely contained in $P$, with the goal of approximating the corresponding optimal profit: The table entry for $(P,W)=(G,W^*)$ will eventually give the desired approximate solution.

If $W\leq \tilde{W}$, for a fixed constant $\tilde{W}$, we simply guess the $W$ edges where the optimum solution puts a weight of one. This provides the optimal profit for drivers in $\mathcal{D}(P)$. Assume now that $W>\tilde{W}$. In this case, by considering all the possible partitions $\overline{P}=\{P_1,\ldots,P_\gamma\}$  of $P$ in $\gamma$ subpaths, we can guess the partition where each $P_i$ takes a $1/\gamma$ fraction of the weight of $P$.
Here $\gamma\geq 2$ is a sufficiently large constant, depending on $\eps$.  Observe that the set $\mathcal{P}_\gamma(P)$ of such partitions has polynomial cardinality. Given $\overline{P}$, for the drivers included into some $P_i$ (i.e., in $\mathcal{D}(P_i)$), we account for the (previously computed) profit of table entry $(P_i,W/\gamma)$.

It remains to consider the profit of drivers $\overline{\mathcal{D}}(P)=\mathcal{D}(P)-\cup_{i=1}^{\gamma}\mathcal{D}(P_i)$ which are contained in $P$, but not in any $P_i$. This is also the crux of our method. Each driver $D\in \overline{\mathcal{D}}(P)$ consists of a (possibly empty) subset of consecutive subpaths $P_\ell,P_{\ell+1},\ldots,P_r$, plus two (possibly empty) subpaths $P_{left}$ and $P_{right}$, with $P_{left}\subset P_{\ell-1}$ and $P_{right}\subset P_{r+1}$. Observe that, if the budget of $D$ is not exceeded, then each \emph{middle} subpath $P_i$, $i\in \{\ell,\ldots,r\}$, contributes with an additive term $W/\gamma$ to the profit of $D$. In particular, this is independent from the way the weight $W/\gamma$ is distributed along $P_i$. 

The situation is radically different for the boundary subpaths $P_{left}$ and $P_{right}$: for them the profit can range from $0$ to $W/\gamma$, depending on the distribution of the weights along $P_{\ell-1}$ and $P_{r+1}$, respectively. In order to implement efficiently the dynamic programming step, we simply neglect the boundary subpaths. In other terms, we replace $D$ with the \emph{shortened} driver $D^s=D-(P_{left}\cup P_{right})$. At this point, we simply account $(r-\ell+1)W/\gamma$ for the profit of $D$, if this quantity does not exceed its budget, and zero otherwise. This way we obtain the overall profit for drivers in $\overline{\mathcal{D}}(P)$, and hence in $\mathcal{D}(P)$.

This approach has two opposite drawbacks:
\begin{itemize}\itemsep0pt
\item[1.] The profit computed might be too pessimistic. This is because we  do not consider the profit coming from $P_{left}\cup P_{right}$ (in particular, it might be $D=P_{left}\cup P_{right}$, and hence $D^s=\emptyset$).
\item[2.] The profit computed might be too optimistic. In fact, it might happen that the weight along $D^s$ is below the budget of $D$, while the weight along $D$ exceeds it (due to the weight on $P_{left}\cup P_{right}$). In that case we account for a positive profit, while the actual profit is zero.
\end{itemize}
We solve the second problem by restricting our attention to \emph{good} drivers $D\in \overline{\mathcal{D}}(P)$, i.e. drivers which contain $\Omega(1/\eps)$ many subpaths $P_i$. It is then sufficient to scale down all the weights at the end of the process by a factor $1-O(\eps)$ to ensure that the budget of good paths is not exceeded. 

Observe that this does not solve the first problem: indeed, it makes it even worse (since we consider less drivers, besides shortening them). At this point, randomization comes into play. We initially enclose the highway in a \emph{bounding path}. Both the length (i.e., the number of edges) of the bounding path and the position of the highway in it are random variables. To this instance we apply the approach above. Consider a driver $D$ which contributes to the optimal profit. For a proper choice of the random variables, with probability $1-O(\eps)$, $D$ is considered in the dynamic program for a path $P$ of weight $W$ such that the profit of $D$ is much larger than $W/\gamma$. Hence $D$ is good with probability close to one.
This introduces a factor  $1-O(\eps)$ in the approximation ratio.
%

As we will see, the domain of the random variables has polynomial size. Hence, the algorithm above can be easily derandomized by considering all the possible realizations.

We believe that our technique will find other applications, and hence it might be of independent interest. In order to motivate that, we show how to apply it to two related problems (see Section \ref{sec:extensions}):
\begin{itemize}\itemsep0pt
\item[$\bullet$] The \emph{tollbooth} problem is the generalization of the highway problem where the input graph $G$ is a tree (rather than a line). This problem is $\mathbf{APX}$-hard, and the best-known approximation for it $O(\log n/\log\log n)$ \cite{GS10}. Here we present a PTAS for the practically-relevant special case that $G$ has a constant number of leaves.

\item[$\bullet$] In the \emph{maximum-feasibility subsystem} problem we are given a set of vectors $a_1,\ldots, a_m\in \Q^n$ and a set of $m$ pairs $(\ell_j,u_j)$, with $0\leq \ell_j\leq u_j$ and $j=1,\ldots,m$. 
The goal is computing a vector $w\in \Q_{\geq0}^n$ such that the \emph{constraint} $\ell_j\leq a_j^T w\leq u_j$ is satisfied by the largest possible number of indexes $j$. Intuitively, the vectors $a_j^T$ can be interpreted as the rows of a matrix $A$: the product $A\,w\in \Q^m$ is what we wish to upper and lower bound. In this paper we restrict to the case that the vectors $a_j$ have entries in $\{0,1\}$, and the $1$'s appear consecutively (i.e., $A$ is an \emph{interval matrix}).

Elbassioni, Raman, Ray, and Sitters~\cite{ElbassioniRRS-SODA09} show that this problem is $\mathbf{APX}$-hard. Moreover, if we allow a violation of the lower and upper bounds by a factor $(1+\eps)$, then there is a polylogarithmic approximation algorithm running in polynomial time, and an exact algorithm running in quasi-polynomial time\footnote{The latter result is not a contradiction, since we compare to the optimum solution, which may not even slightly violate the inequalities}. Here we show how to obtain a $(1+\eps)$ approximation in polynomial time in the same framework.
\end{itemize}

\subsection{Related Work}\label{sec:related}

The highway problem was even not known to be $\mathbf{NP}$-hard until recently. For example, this is posed as an open problem by Guruswami et al. \cite{GHKKKM05soda}. Weakly $\mathbf{NP}$-hardness was shown via a reduction from \emph{partition} 
by Briest and Krysta \cite{BK06soda}. Very recently, Elbassioni, Raman, and Ray \cite{ERR09corr} proved strongly $\mathbf{NP}$-hardness via a reduction from \emph{max-2-SAT}. Balcan and Blum \cite{BB06ec} give a $O(\log n)$ approximation for the problem. Their algorithm partitions the paths in groups of different length. Then it applies a constant factor approximation algorithm in \cite{GHKKKM05soda} for the \emph{rooted} version of the problem, where all drivers contain a given node, to each group separately. The approximation was very recently improved to $O(\log n/\log\log n)$ by Gamzu and Segev~\cite{GS10}. Their algorithm, which also works for the more general tollbooth problem, combines the notion of tree separators with a generalization of the algorithm for the rooted case mentioned before. The QPTAS by Elbassioni, Sitters, and Zhang \cite{ESZ07esa} exploits the profiling technique introduced by Bansal et al.~\cite{BCES06stoc}. The basic idea is guessing the approximate shape of the cumulative weights to the left and right of a given edge. This allows one to partition the problem into two sub-problems, which can be solved recursively. 

There are better approximation results, all based on dynamic programming, for a number of special cases.
In \cite{BB06ec} a constant approximation is given for the case that all the paths have roughly the same length. 
An FPTAS is described by Hartline and Koltun \cite{HK05wads} for the case that the highway has constant length (i.e., $n=O(1)$). This was generalized to the case of constant-length paths in \cite{GHKKKM05soda}. In \cite{GHKKKM05soda} the authors also present an FPTAS for the case that budgets are upper bounded by a constant. An FPTAS is also known \cite{BB06ec,BK06soda} for the case that paths induce a laminar family\footnote{In a laminar family of paths, two paths which intersect are contained one in the other.}.  

The \emph{tollbooth} problem is the generalization of the highway problem where $G$ is a tree. A $O(\log n)$ approximation was developed in \cite{ERR09corr}. As already mentioned, this was very recently improved to $O(\log n/\log \log n)$~\cite{GS10}. The tollbooth problem is $\mathbf{APX}$-hard \cite{GHKKKM05soda}, and 
for general graphs it is $\mathbf{APX}$-hard even when the graph has bounded degree, the paths have constant length and each edge belongs to a constant number of paths \cite{BK06soda}.

The highway and tollbooth problems belong to the family of prizing problems with single-minded customers and unlimited supply. Here we are given a set of customers: Each customer wants to buy a subset of items (\emph{bundle}), if its total prize does not exceed her budget. In the highway terminology, each driver is a subset of edges (rather than a path).
For this problem a $O(\log n+\log m)$ approximation is given in \cite{GHKKKM05soda}. This bound was refined in~\cite{BK06soda} to $O(\log L+\log B)$, where $L$ denotes the maximum number of items in a bundle and $B$ the maximum number of bundles containing a given item. A $O(L)$ approximation is given in \cite{BB06ec}. On the negative side, Demaine et al. \cite{DFHS06soda} show that this problem is hard to approximate within $\log^d n$, for some $d>0$, assuming that $\mathbf{NP}\not\subseteq \mathbf{BPTIME}(2^{n^{\eps}})$ for some $\eps>0$.

The highway problem has some aspects in common with the well-studied \emph{unsplittable flow} problem on line graphs. In this problem we are given a line graph $G=(V,E)$, with edge capacities and a set of paths $D_j$, each one characterized by a demand and a profit. The goal is selecting a maximum profit subset of paths such that the sum of the demands of selected paths on each edge does not exceed the corresponding capacity. For the special case of unit capacities and demands, a $(2+\eps)$ approximation is given by Calinescu et al. \cite{CCKR02ipco}, improving on \cite{BBFNS01jacm,PUW00soda}. Under the \emph{no-bottleneck} assumption, the same approximation guarantee is achieved for the general case by Chekuri, Mydlarz, and Shepherd \cite{CMS07talg}, improving on an earlier constant approximation under the same assumption \cite{CCGK07algo}. Eventually, a QPTAS is given in \cite{BCES06stoc}. The QPTAS for the highway problem in \cite{ESZ07esa} exploits the same basic technique as in \cite{BCES06stoc}. Our hope is that, in turn, our PTAS for the highway problem will inspire a PTAS for the line-graph unsplittable flow problem. However, this seems to require some new ideas and we leave it as a challenging open problem.

For general $0/1$-matrices, the maximum-feasible subsystem problem (with no violation) is not approximable within $\Omega(n^{1/3-\eps})$ for any $\eps>0$ even for $\ell_j = u_j$, unless $\mathbf{ZPP}=\mathbf{NP}$~\cite{ElbassioniRRS-SODA09}. If each row of $A$ contains 
3 non-zero arbitrary coefficients, then even $n^{1-\eps}$ approximations are not possible
in polynomial time~\cite{GuruswamiRaghavendraSTOC07} (see also the previous hardness result \cite{FeigeMaxFShardness97}). The best-known  $O(n / \log n)$ approximation for this problem is due to Halld{\'o}rsson~\cite{Halldorsson00ApproximationsWeightedIS}.

The technique behind our PTAS resembles Arora's quadtree dissection for Euclidean network design \cite{A98jacm}. The basic idea there is enclosing the set of input points into a bounding box, then recursively partition it in a constant number of boxes. This dissection is then randomly shifted. On the resulting random dissection, one applies dynamic programming. We similarly enclose the highway in a bounding path, and partition the latter. Like in Arora's approach, the dissection is randomly shifted. Differently from that case and crucially for our analysis, the size of the bounding path is a random variable as well. Another major difference is that the dissection is not uniform with respect to input properties, but with respect to the optimal weights: for this reason the dissection is constructed in a bottom-up, rather than top-down, fashion via dynamic programming (while computing the approximate solution in parallel). 

\subsection{Preliminaries}\label{sec:preliminaries}

Let $OPT=(w^*,\mathcal{D}^*)$ be the optimum solution, where $w^*$ is the optimal weight function and $\mathcal{D}^*$ is the set of drivers $D_j$ such that $w^*(D_j)\leq \budget_j$. By $opt$ we denote the optimal profit. Our PTAS starts with a sequence of rounding steps to transform the input (and the optimum  solution) in a convenient form, while losing only a factor $1-2\eps$ in the approximation. Since these steps are rather standard, we discuss them here, while in Section \ref{sec:ptas} we will focus on the novel techniques introduced in this paper.

W.l.o.g. we assume $1/(2\eps)\in \N$ and $\eps\leq 1/2$. Let $\budget_{\max}$ be the largest budget. After scaling 
all budgets, one has $b_{\max} = m/\eps^2$. Observe that trivially $opt\geq \budget_{\max}$. 
First, we discard all drivers with a budget smaller than $1/\eps$. 
Next, we round down the budgets to the nearest integer. Any solution to the rounded instance
gives a feasible solution of the same value for the original instance. Moreover, the optimal solution to the rounded instance is a good approximation of the original optimum. In fact, $(w,\mathcal{D})$ with $\mathcal{D} := \{D_j\in\mathcal{D}^* \mid b_j \geq 1/\eps\}$ and $w(e) := \frac{w^*(e)}{1+\eps}$ is a feasible solution to the new instance since \vspace{-1mm}
\[
 w(D_j) = \sum_{e\in D_j} \frac{w^*(e)}{1+\eps} \leq \frac{b_j}{1+\eps} \leq \lfloor b_j\rfloor \vspace{-1mm}
\]
for any  $D_j\in\mathcal{D}^*$ with $b_j\geq1/\eps$. The profit of this solution is \vspace{-1mm}
\[
\sum_{D_j\in\mathcal{D}} w(D_j) \geq \sum_{D_j\in\mathcal{D}^*: b_j\geq1/\eps} \frac{w^*(D_j)}{1+\eps} \geq \frac{opt}{1+\eps} - \frac{m}{\eps} 
\geq (1-2\eps)opt. \vspace{-1mm}
\]
As observed in \cite{CCGK07algo}, the optimal weights for this instance can be assumed to be integral. In fact, given the optimal drivers $\mathcal{D}^*$, the corresponding optimal weights $w^*$ can be computed by solving an ILP whose $0$-$1$ constraint matrix is totally unimodular.
Since the largest weight in $w^*$ is trivially not larger than the largest budget (i.e. $m/\eps^{2}$ after rounding), we can conclude that $w^*:E\to \{0,1,\ldots,m/\eps^2\}$. By replacing each edge with a path of length $m/\eps^2$, we can further assume $w^*:E\to \{0,1\}$. Let $W^*=\sum_{e\in E}w^*(e)$ be the total weight of the solution, and $\gamma=(1/\eps)^{1/\eps}$. By adding $W^* \gamma$ dummy edges (not crossed by any driver), say, to the right of the highway, we can assume that $W^*=\gamma^{\ell}$ for some integer $\ell$ (in fact, the weight assigned to dummy edges is irrelevant). Observe that $W^* \leq n\,m\gamma/\eps^2$: hence we can guess the value of $W^*$ in polynomial-time.

We call an instance of the highway problem with the properties above \emph{well-rounded}. The discussion above implies the following lemma. \vspace{-2mm}
\begin{lemma}\label{lem:wellround}
For any $\eps>0$, there is a polynomial reduction from the highway problem to the same problem on well-rounded instances which is approximation-preserving modulo a factor $(1+\eps)$. 
\end{lemma}

\vspace{-2mm}
\section{A PTAS for the Highway Problem}\label{sec:ptas}

From the discussion in Section \ref{sec:preliminaries}, we assume that the input instance is well-rounded. Let $\eps>0$ be a constant parameter, $\delta=1/(2\eps)\in \N$ and $\gamma=(1/\eps)^{1/\eps}$.
Our PTAS {\tt hptas} for the highway problem is described in Figure \ref{fig:ptas}.
\begin{Figure}[th]
\underbar{\bf Input:} Well-rounded highway instance $G=(V,E)$ and $(P_j,\budget_j)$, $j=1,2,\ldots,m$.

\noindent\underbar{\bf Output:} Edge weights $w:E\to \Q_{\geq0}$ 

\noindent\underbar{\bf Algorithm:} \vspace{\inlineshrink}
\begin{itemize*} \itemsep0pt
\item[\bf (B)] {\bf Bounding Phase:}
\begin{itemize}
\item[\bf (B1)] Guess the value of the total weight $W^* = \gamma^\ell$, $\ell\in\N$.
\item[\bf (B2)] Choose integers $x\in \{1,2,\ldots,W^*\}$ and $y\in \{1,2,\ldots,1/\eps\}$ uniformly at random. Attach a path of length $W^*\cdot((1/\eps)^{y}-1)-x$ (resp., $x$) to the right (resp., left) of $G$. Let $G_0$ be the resulting line graph, and $W'=W^*\cdot(1/\eps)^{y}$.
\end{itemize}
\item[\bf (D)] {\bf Dynamic Programming Phase:}
\begin{itemize}
\item[\bf (D1)] For every path $P\subseteq G_0$, \vspace{\inlineshrink}
$$
\phi(P,(1/\eps)^y)=\max_{\substack{w:P\to \{0,1\}\\w(P)=(1/\eps)^y}}\sum_{\substack{D_j\subseteq P,\\ w(D_j)\leq \budget_j}}w(D_j). \vspace{\inlineshrink}
$$
\item[{\bf (D2)}] For every path $P \subseteq G_0$, and for $W=W'/\gamma^q$, $q=\ell-1,\ell-2,\ldots,0$,\vspace{\inlineshrink}
$$
\phi(P,W)=\max_{\overline{P}\in \mathcal{P}_{\gamma}(P)} \Bigg\{ \sum_{i=1}^{\gamma} \phi\left(P_i,W/\gamma\right)  + \sum_{\substack{D_j\subseteq P,\\ n_j:=| \{ i : P_i \subseteq D_j\}|\geq \delta,\\ W/\gamma\cdot n_j\leq \budget_j}} W/\gamma\cdot n_j   \Bigg\}.\vspace{\inlineshrink}
$$
\end{itemize}
\item[\bf (S)] {\bf Scaling Phase:}
\begin{itemize}
\item[\bf (S1)] Derive $w':G_0\to \{0,1\}$ determining the value of $\phi(G_0,W')$.
\item[\bf (S2)] Output $w:E\to \Q_{\geq 0}$, where $w(e)=w'(e)\cdot \frac{\delta}{\delta+2}$. \vspace{\inlineshrink}
\end{itemize}
\end{itemize*}
\caption{PTAS for the highway problem. Here $\delta:=1/(2\eps)\in \N$ and $\gamma=(1/\eps)^{1/\eps}$.} \label{fig:ptas}
\end{Figure}

In the \emph{Bounding Phase} (B), we first guess the total optimal weight $W^*$ (Step B1). By guessing, we mean that we run the rest of the algorithm for every feasible choice of $W^*$ (which is a polynomially bounded integer). Then, we enclose the highway in a bounding path (Step B2). Both the length of the bounding path and the position of the highway are proper functions of two random variables $x$ and $y$. All the probabilities and expectations in this paper are with respect to the choice of those two variables.

In the \emph{Dynamic Programming Phase} (D), we compute the almost optimal profit $\phi(P,W)$ which can be obtained from the drivers in $P$ by placing $W$-many $1$'s along $P$. In the initialization step (Step D1), we compute profits $\phi(P,(1/\eps)^y)$ by brute force, considering all the $\binom{|P|}{(1/\eps)^y}$-many possible ways to place $(1/\eps)^y=O(1)$-many $1$'s on the edges of $P$.
In the dynamic programming step (Step D2), we consider the \emph{best} partition $\overline{P}=\{P_1,\ldots,P_\gamma\}$ of $P$ into $\gamma$ subpaths. The set of candidate partitions is denoted by $\mathcal{P}_{\gamma}(P)$. We first add to $\phi(P,W)$ the profits $\phi(P_i,W/\gamma)$ for each $i$. Then we consider the good drivers $D_j$, i.e. the drivers in $P$ which contain $n_j\geq \delta$ subpaths $P_i$. For each such driver, we increase $\phi(P,W)$ by the profit associated to the shortened driver $D^s_j=\cup_{P_i\subseteq D_j}P_i$, i.e. $W/\gamma\cdot n_j$, unless this quantity exceeds the budget $b_j$.

In the final \emph{Scaling Phase} (S), we derive from the dynamic programming table the weights $w'$ determining the value of $\phi(G_0,W')$ (Step S1). Then we restrict our attention to the edges of the (original) highway, and scale the corresponding weights down by $\frac{\delta}{\delta+2}$ (Step S2).

\section{Analysis}\label{sec:analysis}

To avoid any confusion, let $n$ and $\bar{n}$ denote the number of edges in the original and well-rounded instance, respectively. Recall that, for any constant $\eps$,  $\bar{n}$ is polynomially bounded in $n$ and $m$. 
\begin{lemma}\label{lem:time}
Algorithm {\tt hptas} runs in polynomial time.
\end{lemma}
\begin{proof}
Since $W^*$ is an integer bounded by $nm\gamma/\eps^2$, its value can be guessed by trying a polynomial number of values. For all the $O(\bar{n}^2)$ choices of $P$ in Step D1, the number of candidate functions $w$ to be considered is $O(\bar{n}^{(1/\eps)^y})$. In Step D2, for all the $O(\bar{n}^2)$ choices of $P$, there are $O(\bar{n}^{\gamma-1})$ possible choices for the $P_i$'s. The claim follows.
\end{proof}

In the rest of the analysis we consider only the run of the algorithm where $W^*$ is guessed correctly. 
The next lemma shows that the profit $apx$ of the finally returned solution, essentially coincides
with the value  $apx_D=\phi(G_0,W')$, that we obtain by dynamic programming.  Here we crucially exploit the fact that we only consider (good) drivers $D_j$ with large $n_j$. 

\begin{lemma}\label{lem:scaling}
$apx\geq \frac{1}{1+4\eps}apx_D$.
\end{lemma}
\begin{proof}
Let $w'$ and $\mathcal{D}'$ be the weights and the set of drivers determining $apx_D$. Consider the corresponding dissection, and let $n_j=|\{i : P_i\subseteq D_j\}|$ and $D^s_j=\bigcup_{P_i\subseteq D_j}P_i$ be defined with respect to that dissection for each $D_j$.  

For any $D_j\in \mathcal{D}'$, $n_j\geq \delta=1/(2\eps)$ and $w'(D^s_j)=W/\gamma\cdot n_j\leq \budget_j$. 
The difference in weight between $D_j$ and $D_j^s$ only
lies in the two sub-intervals owning the endings of $D_j$, and hence $w'(D^s_j)\leq w'(D_j)\leq \frac{W}{\gamma}(n_j+2)$. It follows that 
$
w(D_j)=\frac{\delta}{\delta+2}w'(D_j)\leq \frac{n_j}{n_j+2}w'(D_j)\leq n_j\frac{W}{\gamma}\leq \budget_j. 
$
Hence, $D_j$ contributes to $apx$ with a profit $w(D_j)\geq \frac{\delta}{\delta+2}w'(D^s_j)=\frac{1}{1+4\eps}w'(D^s_j)$. The claim follows since $apx\geq \sum_{D_j\in \mathcal{D}'}w(D_j)\geq \frac{1}{1+4\eps}\sum_{D_j\in \mathcal{D'}}w'(D^s_j)=\frac{1}{1+4\eps}apx_D$.
\end{proof}

It remains to lower bound $apx_D$ in terms of $opt$. In order to simplify the analysis, suppose that we are given an oracle which, for a given $P\; \subseteq G_0$ with $w^*(P)=W=W'/\gamma^q$, $q<\ell$, produces a partition $\overline{P}^*=\{P^*_1,\ldots,P^*_{\gamma}\}$ such that $w^*(P^*_i)=W/\gamma$. Also assume that we remove all the drivers but the ones $\mathcal{D}^*$ in the optimal solution. Consider the variant of Step D where we apply recursively the following Bellman equation \vspace{\inlineshrink}
\begin{equation*}
\phi'(P,W)= \sum_{i=1}^{\gamma} \phi'\left(P^*_i,W/\gamma\right)  + \sum_{\substack{\mathcal{D}^*\ni D_j\subseteq P,\\ n_j:=| \{ i : P^*_i \subseteq D_j\}|\geq \delta,\\ W/\gamma\cdot n_j\leq \budget_j}} W/\gamma\cdot n_j, \vspace{\inlineshrink}
\end{equation*}
until $W=(1/\eps)^y$, in which case we use brute force to compute the optimal weights like in Step D1. It is not hard to see that $apx_O:=\phi'(G_0,W')$ is a lower bound on $apx_D$.
\begin{corollary}\label{lem:comparison}
$apx_D\geq apx_O.$
\end{corollary}
Hence it is sufficient to lower bound $apx_O$. The value $apx_O$ is associated to a unique \emph{optimal dissection}. With the same notation as in the proof of Lemma \ref{lem:scaling}, we let, for a given driver $D_j$, $n_j$ and $D^s_j$ be defined with respect to the optimal dissection. We next say that a subpath in the optimal dissection is \emph{at level} $q\in \{0,1,\ldots,\ell\}$ if its optimal weight is $W'/\gamma^q$. Similarly, we say that a driver $D_j$ is at level $q$ in the optimal dissection if it is contained in a subpath of level $q$, but not $q+1$.

Let $\alpha_q=W'/\gamma^q$. Consider any driver $D_j\in \mathcal{D}^*$, with $\alpha_{q+1}< w^*(D_j)\leq \alpha_q$. We call $D_j$ \emph{good} if it is at level $\ell$ in the dissection, or it is at level $q<\ell$ and it contains at least $\delta$ subpaths of level $q+1$ (i.e., $n_j\geq \delta$).

Observe that good drivers $D_j$ contribute to the value of $apx_O$ with a profit $w^*(D^s_j)\geq w^*(D_j)\cdot \frac{\delta}{\delta+2}=\frac{1}{1+4\eps}\cdot  w^*(D_j)$. 
Hence, it is sufficient to show that a given driver in $\mathcal{D}^*$ is good with probability close to one.
\begin{lemma} \label{lem:good}
Each driver $D_j\in \mathcal{D}^*$ is good with probability at least $1-3\eps$. 
\end{lemma}
\begin{proof}
Let us upper bound the probability that a driver $D_j$ is {\em bad} (i.e., not good). We say that driver $D_j$ is \emph{risky} if 
\[
 \exists q: \eps\,\alpha_q < w^*(D_j) < \frac{1}{\eps}\alpha_q. 
\]
Consider a  log-scale axis and term \emph{tick} the distance that corresponds to a factor of $1/\eps$. Then consecutive  $\alpha_q$'s have a distance of $1/\eps$ ticks to each other (see Figure~\ref{fig:LogscaleAxis}). 
\begin{Figure}
\begin{center}
\psset{xunit=0.7cm, yunit=1cm}
\begin{pspicture}(-1,-1)(23,0.8)
  \pspolygon[fillstyle=solid, fillcolor=lightgray, linestyle=none](2,-10pt)(2,10pt)(4,10pt)(4,-10pt)
  \pspolygon[fillstyle=solid, fillcolor=lightgray, linestyle=none](6,-10pt)(6,10pt)(8,10pt)(8,-10pt)
  \pspolygon[fillstyle=solid, fillcolor=lightgray, linestyle=none](10,-10pt)(10,10pt)(12,10pt)(12,-10pt)
  \pspolygon[fillstyle=solid, fillcolor=lightgray, linestyle=none](19,-10pt)(19,10pt)(21,10pt)(21,-10pt)
  \psaxes[labels=none, linestyle=none, linecolor=white]{->}(0,0)(0,0)(23,0) 
  \psline[linecolor=black]{|->}(0,0)(23,0)
  \psline(3,-10pt)(3,10pt) \rput[c](3,15pt){$\alpha_{q+1}$}
  \psline(7,-10pt)(7,10pt) \rput[c](7,15pt){$\alpha_{q}$}
  \psline(11,-10pt)(11,10pt) \rput[c](11,15pt){$\alpha_{q-1}$}
  \psline(17,-10pt)(17,10pt) \rput[c](17,15pt){$W^*$}
  \psline(20,-10pt)(20,10pt) \rput[c](20,15pt){$\alpha_0 = W'$}
  \psline[arrowsize=6pt]{|<->|}(14,-0.7)(15,-0.7) \rput[c](14.5,-1.05){$1$ tick}
  \psline[arrowsize=6pt]{|<->|}(3,-0.7)(7,-0.7) \rput[c](5,-1.05){$1/\eps$ ticks}
  \psline[arrowsize=6pt]{|<->|}(17,-0.7)(20,-0.7) \rput[c](18.5,-1.05){$y$ ticks}
  \rput[c](1,0.5){$\ldots$}
  \rput[c](13,0.5){$\ldots$}
\end{pspicture}
\end{center}
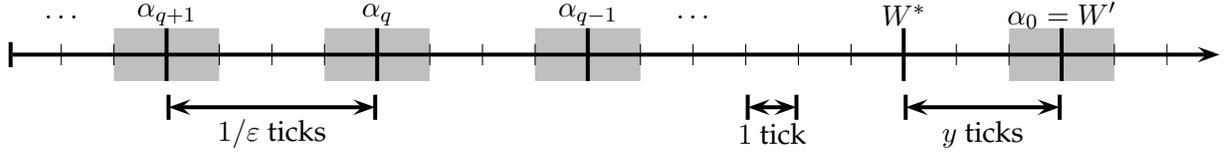
\caption{Log-scale axis. The regions of risky weights are grayshaded.} \label{fig:LogscaleAxis}
\end{Figure}
The region of risky weights w.r.t. a specific $\alpha_q$ is 
the interval $]\eps\cdot\alpha_q, \alpha_q/\eps[$, hence on
the log-scaled axis it is an (open) interval of $2$ ticks length.
The random choice of $y$ yields that all $\alpha_q$'s are simultaneously
shifted by $y\in\{1,\ldots,1/\varepsilon\}$ ticks to the right. Hence for each value of $w^*(D_j)$
at most $2$ out of $1/\eps$ choices of $y$ cause that $D_j$ 
is risky:
\[
\Pr[D_j \textrm{ is risky}] \leq 2\eps. 
\]
Next condition on the event that $D_j$ is not risky. Suppose $D_j$ is not at level $\ell$, otherwise there is nothing to show. Observe that there is a $q$ with \vspace{\inlineshrink}
\[
\frac{1}{\eps} \alpha_q \leq w^*(D_j) \leq \eps\,\alpha_{q-1}. \vspace{\inlineshrink}
\]
Then deterministically $D_j$ contains at least $1/\eps-1\geq1/(2\eps)=\delta$ many level $q$ subpaths. 
Since the random shift $x$ is chosen uniformly at random from
$\{ 1,\ldots,W^* \}$ and $W^*$ is a multiple of $\alpha_{q-1}$\footnote{Except of the case
when $\alpha_{q-1} = W'$, but then deterministically the driver $D_j$ cannot cross the boundary.}
we furthermore have
\[
\Pr[D_j \textrm{ is at level} < q] \leq \frac{w^*(D_j)}{a_{q-1}} \leq \eps. 
\]
Applying the union bound, we obtain that driver $D_j$ is bad with
probability at most $3\eps$.
\end{proof}
\begin{corollary}\label{cor:good}
$E[apx_O]\geq \frac{1-3\eps}{1+4\eps}opt.$
\end{corollary}
\begin{proof}
By linearity of expectation \vspace{\inlineshrink}
\begin{align*}
\hspace{-3.5mm} E[apx_O] & \geq E\Big[\sum_{\substack{D_j\in \mathcal{D}^*,\\D_j\text{ good}}}w^*(D^s_j)\Big] \geq\frac{1}{1+4\eps} E\Big[\sum_{\substack{D_j\in \mathcal{D}^*,\\D_j\text{ good}}}w^*(D_j)\Big] \geq \frac{1-3\eps}{1+4\eps}\sum_{D_j\in \mathcal{D}^*}w^*(D_j)=\frac{1-3\eps}{1+4\eps}opt. \;\; \qedhere 
\end{align*}
\end{proof}

\noindent Now we have all the ingredients to prove the main result of this paper.
\begin{theorem}\label{thr:ptas}
There is a randomized PTAS for the highway problem.
\end{theorem}
\begin{proof}
Consider the randomized algorithm which first transforms the input in a well-rounded instance as described in Section \ref{sec:preliminaries}, and then applies algorithm {\tt hptas}. From Lemmas \ref{lem:wellround} and \ref{lem:time}, this algorithm takes polynomial time. By Lemma \ref{lem:wellround}, Lemma \ref{lem:scaling}, Lemma \ref{lem:comparison}, and Corollary \ref{cor:good}, the approximation ratio of the algorithm is $\frac{(1+4\eps)^2(1+\eps)}{1-3\eps}$.
\end{proof} 
\noindent The PTAS in Theorem \ref{thr:ptas} can be derandomized by considering all the (polynomially many) choices of $x$ and $y$ in Step B2.
\begin{corollary}\label{cor:ptas}
There is a deterministic PTAS for the highway problem.
\end{corollary}

\section{Extensions}\label{sec:extensions}

In this section we extend our approach to two variants of the highway problem.

\subsection{Tollbooth with a Constant Number of Leaves}

We next sketch a PTAS for the tollbooth problem, when the input graph $G$ is a tree with a constant number $\theta=O(1)$ of leaves: details are given in Appendix \ref{appendix:TollboothConstantNumOfLeaves}. Recall that the problem is $\mathbf{APX}$-hard when the number of leaves is arbitrary.

By the same arguments as in the highway case, we assume that optimal weights $w^*$ are $0/1$-valued, and that their sum $W^*$ is bounded by a polynomial in $n$. We choose an arbitrary leaf $s(G)$ of $G$ as a \emph{source}, and call the other leaves \emph{sinks}. Analogously, given any subtree $T$ of $G$, we call the leaf $s(T)$ of $T$ which is closest to $s(G)$, the \emph{source} of $T$. The other leaves of $T$ are called \emph{sinks} of $T$. By appending a path of length $W^*\gamma$ to $s(G)$, we can assume that the total weight along each source-sink pair is $\tilde{W}:=\gamma^\ell$ for some integer $\ell$. The resulting instance is \emph{well-rounded}.

Imagine to split $G$ at any node whose $w^*$-distance from $s(G)$ is
an integer multiple of $\tilde{W}/\gamma$. In such a way we
obtain a forest $\overline{T}=\{T_1,\ldots,T_q\}$ of subtrees with the following property:
any source-sink path in $T_i$ has weight $\tilde{W}/\gamma$. We iterate
this process until the total weight which has to be installed on
the subtree reaches a constant value. We call this dissection \emph{optimal}. 

Consider a driver $D_j$ and let $T$ be the smallest subtree in the optimal dissection 
that fully contains $D_j$. Suppose $W=\tilde{W}/\gamma^q$ is the weight that $w^*$
installs on any source-sink path of $T$. Let $\overline{T} = \{ T_1,\ldots,T_q\}$ be the partition of $T$ in the optimal dissection.
We say that $D_j$ \emph{crosses} $T_i$ if it contains exactly one source-sink path of $T_i$.
We say that driver $D_j$ is \emph{good} if the number $n_j$ of
crossed subtrees is at least a large constant $\delta := \frac{1}{2\eps}$. Also in this case, we can define a shortened driver $D_j^s = \bigcup_{T_i \textrm{ crossed by }D_j} (T_i \cap D_j)$. 
However note that in this case $D^s_j$ might consist of two disjoint paths. (In particular, this might happen if $D_j$ does not lie along a source-sink path of $G$).

Analogously to the highway case, it is sufficient to show that the profit coming from shortened drivers is large with respect to the optimal dissection. Then for subtrees $T$ of the instance and weights $W$, 
we compute table entries $\phi(T,W)$ giving the optimum profit 
that can be obtained from the shortened
paths of good drivers $D_j \subseteq T$, in such a way that on each path from $s(T)$ to any other
leaf of $T$ one installs a total weight of $W$. 

\begin{theorem}\label{thr:tollbooth}
There is a deterministic PTAS for the tollbooth problem with a constant number of leaves.
\end{theorem}

\subsection{Maximum-Feasible Subsystem for Interval Matrices}

In this section we sketch a multi-criteria PTAS for the maximum-feasible subsystem problem on interval matrices ($\maxfs$). More precisely, we show the slightly more general statement:
\begin{theorem} \label{thm:MaxFSGeneralPTAS}
Given a matrix $A \in \{0,1\}^{m\times n}$ with rows $a_1,\ldots,a_m$ having consecutive ones, 
weights $v_1,\ldots,v_m\in\Q_{\geq0}$ and integer bounds $0 \leq \ell_j \leq u_j$, $j=1,\ldots,m$. Let $opt = \max_{w\geq\mathbf{0}} \{ \sum_{j: \ell_j \leq a_j^Tw \leq u_j} v_j \}$.
Then for every fixed $\varepsilon > 0$ one can compute deterministically in polynomial time in $n$, $m$ and $\log \max \{ \ell_j \}$, 
a weight function $w\geq\mathbf{0}$ and a set $J \subseteq \{ 1,\ldots,m \}$ such that $\sum_{j\in J} v_j \geq (1-\varepsilon)opt$
and $\ell_j \leq a_j^Tw \leq (1+\varepsilon)u_j$ for all $j\in J$.
\end{theorem}
By standard arguments, one can round the profits $v_j$ such that
they become integers between $0$ and $m/\varepsilon$. Then each constraint $j$
can be replaced by $v_j$ many constraints with unit profit. 
Choosing $\varepsilon$ accordingly smaller and scaling the weight function by $1 + O(\varepsilon)$,
it suffices to find a solution $w$ that satisfies 
$opt/(1+O(\varepsilon))$ many constraints approximately, i.e.
$\ell_j/(1 + O(\varepsilon)) \leq a_j^Tw \leq u_j(1 + O(\varepsilon))$.

It is maybe easier to think of $\maxfs$ as a variant of the highway problem where: (1) the consecutive $1$'s in each row $j$ define a driver $D_j$ in a line graph $G=(V,E)$, (2) each driver $D_j$, besides having a budget $\budget_j=u_j$, also has a minimum amount of money $\ell_j$ that she wants to spend, and (3) the goal now is maximizing the number of \emph{satisfied} drivers who take the highway (rather than maximizing the profit). Here $w$ can be interpreted as a vector of weights.

Let $OPT = (w^*, \mathcal{D}^*)$ be the optimal solution and define $W^* := \sum_{e\in E} w^*(e)$. 
Abbreviate $\ell_{\max}:=\max\{ \ell_j \mid j=1,\ldots,m\}$. Observe that w.l.o.g. $w^*(e) \leq \ell_{\max}$ 
on all edges. Hence, $W^* \leq n\cdot\ell_{\max}$. 
Since interval matrices are totally unimodular, we can also assume 
that $w^*(e) \in \Z_{\geq 0}$ for all $e\in E$. 
By adding a dummy edge to the left of the line graph (i.e., a zero column 
to the left of the matrix), we can assume that $W^* = (1/\eps)^{\ell/\eps}$ for some $\ell\in\N$. 
We also attach a dummy edge to the right of the graph.


Furthermore recall that for the highway PTAS we duplicate edges in order to obtain $0/1$ weights. The goal is guaranteeing that we can partition the total optimal weight in $\gamma^i$ pieces, $i=1,\ldots,\ell$, without splitting any edge. This is not possible here due to the fact that optimal edge weights are not necessarily polynomially bounded. However, it is sufficient to duplicate each edge $\gamma \cdot \ell \cdot m$ times to achieve the same goal\footnote{The same approach can be used in the highway problem as well, though it is not crucial to obtain a polynomial running time in that case.} (see Appendix~\ref{appendix:MaxFS} for a proof).
Altogether, we obtain a \emph{well-rounded} instance $G_0$ with the following properties:
(1) between any two nodes that are starting point or end point of some driver, one has
at least $\gamma\cdot\ell\cdot m$ edges; (2) the weight of the optimal solution
is a power of $(1/\varepsilon)^{1/\varepsilon}$; (3) at both endings of the highway we have $\gamma\cdot\ell\cdot m$
many edges that are not used by any driver.

Our algorithm applies for such well-rounded instances and begins by
guessing $W^*$. Since $W^*$ is a power of $(1/\eps)^{1/\eps}$, 
there are at most a polynomial number of candidate values.
Recall that the randomization in the algorithms before was used to create a new
probabilistic optimal solution. The careful reader might have noticed that 
the random choice of $x$ can also be moved to the analysis. 
To simplify a later derandomization, in the algorithm we only choose  $y\in\{ 1,\ldots, 1/\eps\}$ uniformly at random and approximate a solution that 
installs a total weight of $W' = (1/\varepsilon)^y\cdot W^*$ on the edges.


For any subpath $P \subseteq G_0$, we compute table entries $\phi(P,W)$ over all weight assignments $w:P \to \Z_{\geq 0}$, with $w(P) = W$, and over all possible dissections of $P$, with the goal of maximizing the number of drivers $D_j$ such that: (1) $D_j$ is fully contained in $P$, (2) $D_j$ is good in the same sense as in the highway case, and (3) $\ell_j/(1+4\eps) \leq w(D_j^s) \leq u_j$ (the shortened driver is \emph{approximately satisfied}). 
The number of such table entries is bounded by a polynomial in $n, m$ and $\log \ell_{\max}$, since we only consider values $W$, which are of the form $W' / \gamma^i$.

Eventually we output the solution $(w,\mathcal{D}')$ that attains the value $\phi(G_0,W')$. 
Using the arguments in Lemma~\ref{lem:good} and Lemma~\ref{cor:good}, 
one can show that $E[\phi(G_0,W')] \geq (1-3\varepsilon)opt$. Similar to 
Lemma~\ref{lem:scaling}, one has $\ell_j/(1+4\varepsilon) \leq w(D_j) \leq u_j(1+4\varepsilon)$ for any $D_j \in \mathcal{D}'$.
Theorem~\ref{thm:MaxFSGeneralPTAS} then follows (see Appendix~\ref{appendix:MaxFS} for more details).

\newpage


\begin{thebibliography}{10}

\bibitem{A98jacm}
S.~Arora.
\newblock Polynomial time approximation schemes for euclidean traveling
  salesman and other geometric problems.
\newblock {\em J. ACM}, 45(5):753--782, 1998.

\bibitem{BB06ec}
M.-F. Balcan and A.~Blum.
\newblock Approximation algorithms and online mechanisms for item pricing.
\newblock In {\em ACM Conference on Electronic Commerce}, pages 29--35, 2006.

\bibitem{BCES06stoc}
N.~Bansal, A.~Chakrabarti, A.~Epstein, and B.~Schieber.
\newblock A quasi-{PTAS} for unsplittable flow on line graphs.
\newblock In {\em STOC}, pages 721--729, 2006.

\bibitem{BBFNS01jacm}
A.~Bar-Noy, R.~Bar-Yehuda, A.~Freund, J.~Naor, and B.~Schieber.
\newblock A unified approach to approximating resource allocation and
  scheduling.
\newblock {\em J. ACM}, 48(5):1069--1090, 2001.

\bibitem{BK06soda}
P.~Briest and P.~Krysta.
\newblock Single-minded unlimited supply pricing on sparse instances.
\newblock In {\em SODA}, pages 1093--1102, 2006.

\bibitem{CCKR02ipco}
G.~Calinescu, A.~Chakrabarti, H.~J. Karloff, and Y.~Rabani.
\newblock Improved approximation algorithms for resource allocation.
\newblock In {\em IPCO}, pages 401--414, 2002.

\bibitem{CCGK07algo}
A.~Chakrabarti, C.~Chekuri, A.~Gupta, and A.~Kumar.
\newblock Approximation algorithms for the unsplittable flow problem.
\newblock {\em Algorithmica}, 47(1):53--78, 2007.

\bibitem{CMS07talg}
C.~Chekuri, M.~Mydlarz, and F.~B. Shepherd.
\newblock Multicommodity demand flow in a tree and packing integer programs.
\newblock {\em ACM Transactions on Algorithms}, 3(3), 2007.

\bibitem{DFHS06soda}
E.~D. Demaine, U.~Feige, M.~Hajiaghayi, and M.~R. Salavatipour.
\newblock Combination can be hard: approximability of the unique coverage
  problem.
\newblock In {\em SODA}, pages 162--171, 2006.

\bibitem{ERR09corr}
K.~M. Elbassioni, R.~Raman, and S.~Ray.
\newblock On profit-maximizing pricing for the highway and tollbooth problems.
\newblock {\em CoRR}, abs/0901.1140, 2009.

\bibitem{ElbassioniRRS-SODA09}
K.~M. Elbassioni, R.~Raman, S.~Ray, and R.~Sitters.
\newblock On the approximability of the maximum feasible subsystem problem with
  0/1-coefficients.
\newblock In {\em SODA}, pages 1210--1219, 2009.

\bibitem{ESZ07esa}
K.~M. Elbassioni, R.~Sitters, and Y.~Zhang.
\newblock A quasi-ptas for profit-maximizing pricing on line graphs.
\newblock In {\em ESA}, pages 451--462, 2007.

\bibitem{FeigeMaxFShardness97}
U.~Feige and D.~Reichman.
\newblock On the hardness of approximating max-satisfy.
\newblock {\em Information Processing Letters}, 97(1):31 -- 35, 2006.

\bibitem{GS10}
I.~Gamzu and D.~Segev.
\newblock A sublogarithmic approximation for highway and tollbooth pricing.
\newblock {\em CoRR}, abs/1002.2084, 2010.

\bibitem{GHKKKM05soda}
V.~Guruswami, J.~D. Hartline, A.~R. Karlin, D.~Kempe, C.~Kenyon, and
  F.~McSherry.
\newblock On profit-maximizing envy-free pricing.
\newblock In {\em SODA}, pages 1164--1173, 2005.

\bibitem{GuruswamiRaghavendraSTOC07}
V.~Guruswami and P.~Raghavendra.
\newblock A 3-query pcp over integers.
\newblock In {\em STOC}, pages 198--206, New York, 2007.

\bibitem{Halldorsson00ApproximationsWeightedIS}
M.~M. Halldorsson.
\newblock Approximations of weighted independent set and hereditary subset
  problems.
\newblock {\em Journal of Graph Algorithms and Applications}, 4:2000, 2000.

\bibitem{HK05wads}
J.~D. Hartline and V.~Koltun.
\newblock Near-optimal pricing in near-linear time.
\newblock In {\em WADS}, pages 422--431, 2005.

\bibitem{PUW00soda}
C.~A. Phillips, R.~N. Uma, and J.~Wein.
\newblock Off-line admission control for general scheduling problems.
\newblock In {\em SODA}, pages 879--888, 2000.

\end{thebibliography}


\appendix

\section{Tollbooth with a Constant Number of Leaves\label{appendix:TollboothConstantNumOfLeaves}}

A detailed description of the algorithm is given in Figure \ref{fig:tollbooth}.
By $\mathcal{P}_{\gamma}(T)$ we denote the set of potential dissections of subtree $T$ into a forest $\overline{T}=\{T_1,\ldots,T_{q(\overline{T})}\}$. Observe that each source-sink path of $T$ can contain at most $\gamma-1$ break-points. Consequently, the number $q(\overline{T})$ of subtrees in each candidate forest is at most $\gamma\cdot \theta$. It follows that the cardinality of $\mathcal{P}_{\gamma}(T)$ is polynomially bounded when the number $\theta$ of leaves of $G$ is constant.

\begin{Figure}[th]
\underbar{\bf Input:} Well-rounded tollbooth instance $G=(V,E)$ and $(D_j,\budget_j)$, $j=1,2,\ldots,m$.

\noindent\underbar{\bf Output:} Edge weights $w:E\to \Q_{\geq 0}$.

\noindent\underbar{\bf Algorithm:}
\begin{itemize} \itemsep0pt
\item[\bf (B)] {\bf Bounding Phase:}
\begin{itemize}
\item[\bf (B1)] Guess the value of the weight $\tilde{W} = \gamma^\ell$, $\ell\in\N$.
\item[\bf (B2)] Choose integers $x\in \{1,2,\ldots,\tilde{W}\}$ and $y\in \{1,2,\ldots,1/\eps\}$ uniformly at random. Attach a path of length $\tilde{W}\cdot((1/\eps)^{y}-1)-x$  to each sink of $G$, and a path of length $x$ to the source of $G$. Let $G_0$ be the resulting tree, and $W'=\tilde{W}\cdot(1/\eps)^{y}$.
\end{itemize}
\item[\bf (D)] {\bf Dynamic Programming Phase:}
\begin{itemize}
\item[\bf (D1)] For every subtree $T\subseteq G_0$, 
$$
\phi(T,(1/\eps)^y)=\max_{\substack{w:T\to \{0,1\}\\w(T)=(1/\eps)^y}}\sum_{\substack{D_j\subseteq T,\\ w(D_j)\leq \budget_j}}w(D_j).
$$
\item[{\bf (D2)}] For every subtree $T \subseteq G_0$, $W=W'/\gamma^q$, and $q=\ell-1,\ell-2,\ldots,0$,
$$
\phi(T,W)=\max_{\overline{T}\in \mathcal{P}_{\gamma}(T)} \Bigg\{ \sum_{i=1}^{q(\overline{T})} \phi\left(T_i,W/\gamma\right)  + \sum_{\substack{D_j\subseteq T\\ n_j:=|\{i:\;D_j\text{ crosses }T_i\}|\geq \delta\\W/ \gamma \cdot n_j\leq \budget_j}} W/\gamma\cdot n_j   \Bigg\}
$$
\end{itemize}
\item[\bf (S)] {\bf Scaling Phase:}
\begin{itemize}
\item[\bf (S1)] Derive $w':G_0\to \{0,1\}$ determining the value of $\phi(G_0,W')$.
\item[\bf (S2)] Output $w:E\to \Q_{\geq 0}$, where $w(e)=w'(e)\cdot \frac{\delta}{\delta+4}$. 
\end{itemize}
\end{itemize}
\caption{PTAS for the tollbooth problem for a constant number of leaves. Here $\delta:=1/(2\eps)\in \N$ and $\gamma=(1/\eps)^{1/\eps}$.} \label{fig:tollbooth}
\end{Figure}

\begin{proof} {\em (Theorem \ref{thr:tollbooth})}
Consider the randomized algorithm described in Figure \ref{fig:tollbooth}: this algorithm can be derandomized by considering all the possible values of random variables $x$ and $y$. Assume $0<\eps\leq \frac{1}{8}$ without loss of generality.

Like in the highway case, let us restrict our attention to the dissection corresponding to the optimal weights, and let us discard drivers which do not provide any profit in the optimal solution. 

We start by showing that any residual driver $D_j$ is good with probability at least $1-3\eps$. Let us call a driver $D_j$ \emph{straight} if it lays along a source-sink path of $G$, and \emph{bent} otherwise. By exactly the same argument as in the highway case, a straight path is good with probability at least $(1-3\eps)$. Hence consider a bent driver $D_j$, and let 
$D'_j$ and $D''_j$ be the two straight subpaths which partition $D_j$.
Paths $D_j'$ and $D_j''$ have a common endpoint, which is the node of $D_j$ which is closest to the sink of $G$.  
Without loss of generality, $w^*(D'_j)\geq w^*(D''_j)$. With the same notation as in the highway case, and by a similar argument, with probability at least $1-2\eps$, there is a $q$ such that $\frac{1}{\eps}\alpha_q\leq w^*(D'_j)\leq \eps\alpha_{q-1}$. When this happens, $D_j'$ is at level $q$ in the dissection with probability at least $1-\eps$. Conditioning on the latter event, by the way the dissection is constructed and being $w^*(D''_j)\leq w^*(D'_j)$, $D''_j$ is at level not smaller than $q$ in the dissection. This implies that $D_j$ is at level $q$ as well. We can conclude that $D_j$ crosses at least $\frac{1}{\eps}-4\geq \frac{1}{2\eps}=\delta$ many level $q$ subtrees. The $-4$ here comes from the fact that the portion of $D_j$ not crossing any subtree consists of at most $4$ source-sink subpaths ($2$ for $D'_j$ and $2$ for $D''_j$, if $D_j$ is bent). Altogether, $D_j$ is good with probability at least $1-3\eps$.

Given that $D_j$ is good, the portion of $D_j$ crossing subtrees at level $q+1$ has weight at least $\frac{\delta}{\delta+4}w^*(D_j)$. This is by the same argument as above. Furthermore, the budget of $D_j$ in the dynamic program is violated at most by a factor 
$\frac{\delta+4}{\delta}$: hence scaling the weights by $\frac{\delta}{\delta+4}$ in Step (S2) guarantees that good paths contribute to the actual profit. Considering that the initial rounding introduces a factor $1+\eps$ in the approximation, altogether the solution produced by the algorithm gives profit at least $(\frac{\delta}{\delta+4})^2 \cdot\frac{1-3\eps}{1+\eps}opt=\frac{1-3\eps}{(1+8\eps)^2(1+\eps)}opt$ in expectation.
\end{proof}

\newpage


\section{Maximum-Feasible Subsystem for Interval Matrices}
\label{appendix:MaxFS}

\begin{Figure}[t]
\underbar{\bf Input:} Well-rounded $\maxfs$ instance $G_0=(V,E)$ and $(P_j,\ell_j,u_j)$, $j=1,2,\ldots,m$.

\noindent\underbar{\bf Output:} Edge weights $w:E\to \Z_{\geq 0}$, drivers $\mathcal{D}' \subseteq \mathcal{D}$. 

\noindent\underbar{\bf Algorithm:}
\begin{itemize*} 
  \item[\bf (B)] {\bf Bounding Phase:}
  \begin{itemize}
    \item[\bf (B1)] Guess the value of the total weight $W^* = \gamma^\ell$, $\ell\in\N$.
    \item[\bf (B2)] Choose $y \in \{ 1,\ldots,1/\varepsilon\}$ uniformly at random. Define $W' = W^*\cdot (1/\varepsilon)^y$.
  \end{itemize}
  \item[\bf (D)] {\bf Dynamic Programming Phase:}
  \begin{itemize}
    \item[\bf (D1)] For every path $P\subseteq G_0$, 
\[
\phi(P,(1/\eps)^y)=\max_{\substack{w:P\to \Z_{\geq 0}\\w(P)=(1/\eps)^y}} \big|\big\{ D_j\subseteq P \mid \ell_j/(1+4\varepsilon) \leq w(D_j) \leq u_j \big\}\big|.
\]
For every path $P$ with no $D_j \subseteq P$, define $\phi(P,W) = 0$ for any $W=W'/\gamma^q$, $q=0,\ldots,\ell-1$.
    \item[{\bf (D2)}] For every path $P \subseteq G_0$, and for $W=W'/\gamma^q$, $q=\ell-1,\ell-2,\ldots,0$,
\[
\phi(P,W)=\max_{\overline{P}\in \mathcal{P}_{\gamma}(P)} \Bigg\{ \sum_{i=1}^{\gamma} \phi\left(P_i,W/\gamma\right)  + \Big|\Big\{ D_j\subseteq P \mid \substack{n_j:=| \{ i : P_i \subseteq D_j\}|\geq \delta,\\ \ell_j/(1+4\varepsilon) \leq \frac{W}{\gamma}\cdot n_j \leq u_j} \Big\}\Big|   \Bigg\}.
\]
  \end{itemize}
\item[{\bf (O)}] {\bf Output Phase:}
  \begin{itemize}
  \item[{\bf (O1)}] Derive $w:E\to \Z_{\geq 0}$ and $\mathcal{D}'\subseteq\mathcal{D}$ determining the value of $\phi(G_0,W')$. 
  \item[{\bf (O2)}] Output $(w,\mathcal{D}')$. 
  \end{itemize}
\end{itemize*}
\caption{PTAS with $1+O(\varepsilon)$-violation for the $\maxfs$ problem. Here $\delta:=1/(2\eps)\in \N$ and $\gamma=(1/\eps)^{1/\eps}$.} \label{fig:maxfsptas}
\end{Figure}

Recall that a driver $D_j$ \emph{belongs} to a path $P$ in a dissection, if $P$ is the maximal
path with $D_j \subseteq P$. Suppose the driver $D_j$ indeed belongs to $P$ and 
the dissection splits $P$ into $\overline{P} = \{ P_1,\ldots,P_{\gamma}\}$. Then $D_j$ is termed \emph{good} if
the number of $P_i$'s with $P_i \subseteq D_j$ is at least $\delta = \frac{1}{2\varepsilon}$. 

The algorithm in Figure~\ref{fig:maxfsptas} computes table entries
$\phi(P,W)$ representing the maximum number of good drivers $D_j \subseteq P$ that can be
approximately satisfied under the constraint $w(P) = W$.
The main difference to the previous algorithms is that, if we reach a path $P$
not containing any driver $D_j$, then we define $\phi(P,W) = 0$. 
First note that the number of table entries is bounded by a polynomial in $n$ and $\log \ell_{\max}$. Hence, 
the table entries can be computed in time $poly(n,m,\log \ell_{\max})$.

Next, we argue why the value of the computed table entry is not much worse in expectation
than the optimal number of satisfiable drivers.
\begin{lemma}
The final table entry satisfies $E[\phi(G_0,W')] \geq (1-3\varepsilon)opt$.
\end{lemma}
\begin{proof}
Let $w^* : E \to \Z_{\geq 0}$ be the optimal weight function of total weight $W^*$. 
Recall that we have inserted
dummy edges to the left and to the right, not contained in any driver $D_j$.
We choose an integer $x\in \{1,2,\ldots,W^*\}$ uniformly at random. 
Then increase the total weight on the dummy edges to the right by  $W^*\cdot((1/\eps)^{y}-1)-x$ and the weight on the dummy edges to the left by $x$. The total weight
of $w^*$ is now indeed $W'=W^*\cdot(1/\eps)^{y}$.
It now suffices to show the promised bound on $E[\phi(G_0,W')]$ over the random
choices of $y$ and $x$.

Recall that for $\maxfs$ we could not assume that all edges carry just unit weight.
Hence we need to argue, why there still is a proper dissection induced by $w^*$,
when each edge is replaced by just $\gamma \cdot \ell \cdot m$ many edge segments.
To see this, imagine the line graph $G^*$, which indeed emerges from replacing any edge $e$ 
by $w^*(e)$ many edges. 
As in previous sections, there is a proper dissection induced by $w^*$ --- potentially with an exponential
number of leaves. We think of this dissection to be constructed in a top-down
fashion, where the dynamic program truncates the dissection at \emph{empty} paths, that do not
contain any driver. 
How many paths (or nodes in the dissection tree) can this truncated dissection have?
Any of the $m$ drivers is fully contained in not more than $\ell$ many paths (which is
the depth of the dissection tree). And any remaining empty path must have a father that
is non-empty. Hence the number of paths $P$ in the truncated dissection tree is bounded by 
$\gamma \cdot \ell \cdot m$. Since we replaced any edge in the original graph by that many edge
segments, this truncated dissection also exists in $G_0$.

Again, by Lemma~\ref{lem:good}, if we consider the (truncated) dissection of $G_0$ 
which is induced by the optimal solution, any driver $D_j$ is good with probability at least $1-3\varepsilon$.
Suppose that $D_j$ is good and satisfied in the optimal solution, i.e.  $\ell_j \leq w^*(D_j) \leq u_j$.
Then 
\[
 w^*(D_j^s) \geq \frac{\delta}{\delta+2} w^*(D_j) \geq \ell_j/(1+4\varepsilon)
\]
and of course $w^*(D_j^s) \leq w^*(D_j) \leq u_j$. In other words, $D_j$ would be included by the dynamic program.
The claim follows again by linearity of expectation.
\end{proof}

Finally we argue that the returned drivers are approximately satisfied by the
computed weight function. 
\begin{lemma}
Let $(w,\mathcal{D}')$ be the returned solution.
For every driver $D_j\in\mathcal{D}'$, one has $\ell_j/(1+4\varepsilon) \leq w(P_j) \leq u_j(1+4\varepsilon)$.
\end{lemma}
\begin{proof}
Again let $D_j^s$ be the shortened driver of $D_j$ w.r.t. the dissection induced by the computed 
weight function $w$. First of all $w(D_j) \geq w(D_j^s) \geq \ell_j/(1+4\varepsilon)$. 
Next, the driver $D_j$ is good, hence 
\[
  w(D_j) \leq \frac{\delta+2}{\delta} w(D_j^s) = (1+4\varepsilon)w(D_j^s) \leq (1+4\varepsilon)\cdot u_j.
\]
\end{proof}

We observe that the above algorithm can be easily derandomized by trying
out all $1/\varepsilon$ many choices of $y$. 
In total Theorem~\ref{thm:MaxFSGeneralPTAS} follows. 



\end{document}